\documentclass{article}
\usepackage{authblk}

\usepackage{adjustbox}

\usepackage[normalem]{ulem}

\usepackage{pgfplots}
\usepackage[cmex10]{amsmath}
\usepackage{amsthm}

\newtheorem{theorem}{Theorem}[section]

\newtheorem{lemma}[theorem]{Lemma}

\usepackage{rotating}
\newcommand{\comm}[1]{}

\newcounter{tempEquationCounter} 
\newcounter{thisEquationNumber}
\newenvironment{floatEq}
{\setcounter{thisEquationNumber}{\value{equation}}\addtocounter{equation}{1}
\begin{figure*}[!t]
\normalsize\setcounter{tempEquationCounter}{\value{equation}}
\setcounter{equation}{\value{thisEquationNumber}}
}
{\setcounter{equation}{\value{tempEquationCounter}}
\hrulefill\vspace*{4pt}
\end{figure*}
}

\graphicspath{{diagrams/}}

\begin{document}

\title{ICE: A General and Validated Energy Complexity Model for Multithreaded Algorithms
(IFI-UiT Technical Report 2016-77)}

\author{Vi Ngoc-Nha Tran}
\author{Phuong Hoai Ha}
\affil{The Arctic Green Computing Group\\
Department of Computer Science\\
UiT The Arctic University of Norway\\
Tromso, Norway\\
\{vi.tran, phuong.hoai.ha\}@uit.no}

\date{September 23, 2016}
\maketitle

\thispagestyle{empty}

\begin{abstract}
Like time complexity models that have significantly contributed to the analysis and development of fast algorithms, energy complexity models for parallel algorithms are desired as crucial means to develop energy efficient algorithms for ubiquitous multicore platforms. Ideal energy complexity models should be validated on real multicore platforms and applicable to a wide range of parallel algorithms. However, existing energy complexity models for parallel algorithms are either theoretical without model validation or algorithm-specific without ability to analyze energy complexity for a wide-range of parallel algorithms.  

This paper presents a new general validated energy complexity model for parallel (multithreaded) algorithms. The new model abstracts away possible multicore platforms by their static and dynamic energy of computational operations and data access, and derives the energy complexity of a given algorithm from its {\em work}, {\em span} and {\em I/O} complexity. 
The new model is validated by different sparse matrix vector multiplication (SpMV) algorithms and dense matrix multiplication (matmul) algorithms running on high performance computing (HPC) platforms (e.g., Intel Xeon and Xeon Phi). The new energy complexity model is able to characterize and compare the energy consumption of SpMV and matmul kernels according to three aspects: different algorithms, different input matrix types and different platforms. The prediction of the new model regarding which algorithm consumes more energy with different inputs on different platforms, is confirmed by the experimental results. In order to improve the usability and accuracy of the new model for a wide range of platforms, the platform parameters of ICE model are provided for eleven platforms including HPC, accelerator and embedded platforms.
\end{abstract}

\section{Introduction}
Understanding the energy complexity of algorithms is crucial important to improve the energy efficiency of algorithms \cite{Umar2016, Umar20162, Umar:2015, Kumar:2016:EAM:2959355.2959420} and reduce the energy consumption of computing systems \cite{Tran2016, TranSamos2016, 7568416}. One of the main approaches to understand the energy complexity of algorithms is to devise energy models.
 
Significant efforts have been devoted to developing power and energy models in literature \cite{Alonso2014, Choi2013, Choi2014, Korthikanti2009, Korthikanti2010, 7108419, Mishra:2015, Snowdon:2009}. However, there are no analytic models for multithreaded algorithms that are both applicable to a wide range of algorithms and comprehensively validated yet (cf. Table \ref{table:energy-model-summary}). The existing {\em parallel} energy models are either theoretical studies without validation or only applicable for specific algorithms. Modeling energy consumption of {\em parallel} algorithms is difficult since the energy models must take into account the complexity of both parallel algorithms and parallel platforms. The algorithm complexity results from parallel computation, concurrent memory accesses and inter-process communication. The platform complexity results from multicore architectures with deep memory hierarchy.

The existing models and their classification are summarized in Table \ref{table:energy-model-summary}. 
To the best of our knowledge, the proposed ICE (Ideal Cache Energy) complexity model is the first energy model that covers all three aspects: i) ability to analyze the energy complexity of parallel algorithms (i.e. Energy complexity analysis for parallel algorithms), ii) applicability to a wide range of algorithms (i.e., Algorithm generality), and iii) model validation (i.e., Validation). Section \ref{related-work} describes how the ICE model complements the other currently used models.
\begin{table*}
\caption{Energy Model Summary}
\label{table:energy-model-summary}
\begin{center}
\begin{adjustbox}{width=1\textwidth}
\begin{tabular}{lllll}
\hline\noalign{\smallskip}
\textbf{Study} 	 &\textbf{Energy complexity} 	&\textbf{Algorithm}  	 &\textbf{Validation}  	 \\
	 &\textbf{analysis for} 	&\textbf{ generality}  		&  	\\
	 &\textbf{parallel algorithms} && 				\\
\noalign{\smallskip}\hline\noalign{\smallskip}
LEO \cite{Mishra:2015} 	&No 	&General	&Yes 	\\
POET \cite{7108419} 	&No 	&General	&Yes	\\
Koala \cite{Snowdon:2009} 	&No 	&General	&Yes	\\
Roofline \cite{Choi2013,Choi2014}	&No	&General	&Yes	\\  
Energy scalability \cite{Korthikanti2009, Korthikanti2010}	&Yes	&General	&No	\\ 
Sequential energy complexity \cite{Roy2013}   	&No	&General	&Yes\\
Alonso  et al. \cite{Alonso2014}	&Yes	&Algorithm-specific	&Yes	\\
Malossi  et al. \cite{Malossi2015}	&Yes	&Algorithm-specific	&Yes	\\
\textbf{ICE model (this study)}       	&\textbf{Yes}	&\textbf{General}	&\textbf{Yes}	\\
\noalign{\smallskip}\hline\noalign{\smallskip}
\end{tabular}
\end{adjustbox}
\end{center}
\raggedright{To the best of our knowledge, the ICE model is the first \em{validated} model that supports \em{energy complexity} analysis for \em{general multi-threaded} algorithms.}
\end{table*}

The energy complexity model ICE proposed in this study is for general multithreaded algorithms and validated on three aspects: different algorithms for a given problem, different input types and different platforms. The proposed model is an analytic model which characterizes both algorithms (e.g., representing algorithms by their {\em work}, {\em span} and {\em I/O} complexity) and platforms (e.g., representing platforms by their static and dynamic energy of memory accesses and computational operations). By considering {\em work}, {\em span} and {\em I/O} complexity, the new ICE model is applicable to any multithreaded algorithms. 

The new ICE model is designed for analyzing the energy {\em complexity} of algorithms and therefore the model does not provide the estimation of absolute energy consumption. The goal of the ICE model is to answer energy complexity question: {\em "Given two parallel algorithms A and B for a given problem, which algorithm consumes less energy analytically?"}. Hence, the details of underlying systems (e.g., runtime and architectures) are abstracted away to keep ICE model simple and suitable for complexity analysis. O-notation represents an {\em asymptotic upper-bound} on energy complexity.

In this work, the following contributions have been made.
\begin{itemize}
\item Devising a new general energy model ICE for analyzing the energy complexity of a wide range of multithreaded algorithms based on their {\em work}, {\em span} and {\em I/O} complexity (cf. Section \ref{energy-model}). The new ICE model abstracts away possible {\em multicore platforms} by their static and dynamic energy of computational operations and memory access. The new ICE model complements previous energy models such as energy roofline models \cite{Choi2013, Choi2014} that abstract away possible {\em algorithms} to analyze the energy consumption of different multicore platforms.
\item Conducting two case studies (i.e., SpMV and matmul) to demonstrate how to apply the ICE model to find energy complexity of parallel algorithms. The selected parallel algorithms for SpMV are three algorithms: Compressed Sparse Column(CSC), Compressed Sparse Block(CSB) and Compressed Sparse Row(CSR)(cf. Section \ref{SpMV-energy}). The selected parallel algorithms for matmul are two algorithms: a basic matmul algorithm and a cache-oblivious algorithm (cf. Section \ref{Matmul-energy}). 
\item Validating the ICE energy complexity model with both data-intensive (i.e., SpMV) and computation-intensive (i.e., matmul) algorithms according to three aspects: different algorithms, different input types and different platforms. The results show the precise prediction on which validated SpMV algorithm (i.e., CSB or CSC) consumes more energy when using different matrix input types from Florida matrix collection \cite{Davis:2011} (cf. Section \ref{SpMV-validation}). The results also show the precise prediction on which validated matmul algorithm (i.e., basic or cache-oblivious) consumes more energy (cf. Section \ref{Matmul-validation}). The model platform-related parameters for 11 platforms, including x86, ARM and GPU, are provided to facilitate the deployment of the ICE model.  
\end{itemize}

\section{ICE Shared Memory Machine Model}
\label{EPEM-model}
Generally speaking, the energy consumption of a parallel algorithm is the sum of i) static energy (or leakage) $E_{static}$, ii) dynamic energy of computation $E_{comp}$ and iii) dynamic energy of memory accesses $E_{mem}$. The static energy $E_{static}$ is proportional to the execution time of the algorithm while the dynamic energy of computation and the dynamic energy of memory accesses are proportional to the number of computational operations and the number of memory accesses of the algorithm, respectively \cite{Korthikanti2010}. As a result, in the new ICE complexity model, the energy complexity of a multithreaded algorithm is analyzed based on its {\em span complexity} \cite{CormenLRS:2009} (for the static energy), {\em work complexity} \cite{CormenLRS:2009} (for the dynamic energy of computation) and {\em I/O complexity} (for the dynamic energy of memory accesses) (cf. Section \ref{energy-model}). This section describes shared-memory machine models supporting I/O complexity analysis for parallel algorithms.

The first memory model we consider is parallel external memory (PEM) model \cite{Arge:2008}, an extension of the Parallel Random Access Machine (PRAM) model that includes a two-level memory hierarchy. 
In the PEM model, there are $n$ cores (or processors) each of which has its own {\em private} cache of size $Z$ (in bytes) and shares the main memory with the other cores (cf. Figure \ref{fig:machine-model}). 
When $n$ cores access $n$ distinct blocks from the shared memory {\em simultaneously}, the I/O complexity in the PEM model is $O(1)$ instead of $O(n)$.
\begin{figure}[!t] \centering
\resizebox{0.5\columnwidth}{!}{ \includegraphics{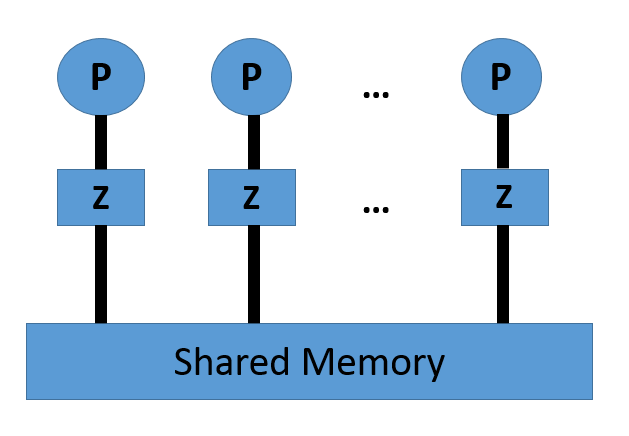}}
\caption{A Shared Memory Machine Model with Private Caches}
\label{fig:machine-model}
\end{figure} 
Although the PEM model is appropriate for analyzing the I/O complexity of parallel algorithms in terms of time performance \cite{Arge:2008}, we have found that the PEM model is not appropriate for analyzing parallel algorithms in terms of the dynamic energy of memory accesses. In fact, even when the $n$ cores can access data from the main memory simultaneously, the {\em dynamic} energy consumption of the access is proportional to the number $n$ of accessing cores (because of the load-store unit activated within each accessing core and the energy compositionality of parallel computations \cite{HaTUTGRWA:2014, chris-eehco14}), rather than a constant as implied by the PEM model.

As a result, we consider the ideal distributed cache (IDC) model \cite{Frigo:2006} to analyze I/O complexity of multithreaded algorithms in terms of dynamic energy consumption. 
Since the cache complexity of $m$ misses is $O(m)$ regardless of whether or not the cache misses are incurred simultaneously by the cores, the IDC model reflects the aforementioned dynamic energy consumption of memory accesses by the cores. 

However, the IDC model is mainly designed for analyzing the cache complexity of divide-and-conquer algorithms, making it difficult to apply to general multi-threaded algorithms targeted by our new ICE model. Constraining the new ICE model to the IDC model would limit the applicability of the ICE model to a wide range of multithreaded algorithms.

In order to make our new ICE complexity model applicable to a wide range of multithreaded algorithms, we show that the cache complexity analysis using the traditional (sequential) ideal cache (IC) model \cite{FrigoLPR:1999} can be used to find an upper bound  on the cache complexity of the same algorithm using the IDC model (cf. Lemma \ref{ePEM}). As the sequential execution of multithreaded algorithms is a valid execution regardless of whether they are divide-or-conquer algorithms, the ability to analyze the cache complexity of multithreaded algorithms via their sequential execution in the ICE complexity model improves the usability of the ICE model.   
 
Let $Q_1(Alg,B,Z)$ and $Q_P(Alg, B, Z)$ be the cache complexity of a parallel algorithm $Alg$ analyzed in the (uniprocessor) ideal cache (IC) model  \cite{FrigoLPR:1999} with block size $B$ and cache size $Z$ (i.e, running $Alg$ with a single core) and the cache complexity analyzed in the (multicore) IDC model with $P$ cores each of which has a private cache of size $Z$ and block size $B$, respectively. We have the following lemma:
\begin{lemma}
\label{ePEM}
The cache complexity $Q_P(Alg, B, Z)$ of a parallel algorithm $Alg$ analyzed in the ideal distributed cache (IDC) model with $P$ cores is bounded from above by the product of $P$ and the cache complexity $Q_1(Alg, B, Z)$ of the same algorithm analyzed in the ideal cache (IC) model. Namely,
\begin{equation}
Q_P(Alg, B, Z) \leq P*Q_1(Alg,B,Z)
\end{equation}
\end{lemma}
\begin{proof} 
(Sketch)
Let $Q_P^{i}(Alg, B, Z)$ be the number of cache misses incurred by core $i$ during the parallel execution of algorithm $Alg$ in the IDC model. Because caches do not interfere with each other in the IDC model, the number of cache misses incurred by core $i$ when executing algorithm $Alg$ in parallel by $P$ cores is not greater than the number of cache misses incurred by core $i$ when executing the whole algorithm $Alg$ only by core $i$. That is,
\begin{equation}
\label{eq:idc_ic_1}
Q_P^{i}{(Alg, B, Z)} \leq Q_1(Alg,B,Z)
\end{equation}
or
\begin{equation}
\label{eq:idc_ic_2}
\sum_{i=1}^P Q_P^{i}{(Alg, B, Z)} \leq P*Q_1(Alg,B,Z)
\end{equation}

On the other hand, since the number of cache misses incurred by algorithm $Alg$ when it is executed by $P$ cores in the IDC model is the sum of the numbers of cache misses incurred by each core during the $Alg$ execution, we have 
\begin{equation}
\label{eq:idc_ic_3}
Q_P(Alg, B, Z) = \sum_{i=1}^P Q_P^{i}{(Alg, B, Z)}
\end{equation}
From Equations \ref{eq:idc_ic_2} and \ref{eq:idc_ic_3}, we have
\begin{equation}
\label{eq:idc_ic_4}
Q_P(Alg, B, Z) \leq P*Q_1(Alg,B,Z)
\end{equation}
\end{proof}

We also make the following assumptions regarding platforms.
\begin{itemize}
\item Algorithms are executed with the best configuration (e.g., maximum number of cores, maximum frequency) following the race-to-halt strategy.
\item The I/O parallelism is bounded from above by the computation parallelism. Namely, each core can issue a memory request only if its previous memory requests have been served. Therefore, the work and span (i.e., critical path) of an algorithm represent the parallelism for both I/O and computation \cite{CormenLRS:2009}. 
\end{itemize}  

\section{Energy Complexity in ICE model}
\label{energy-model}
This section describes two energy complexity models, a platform-supporting energy complexity model considering both platform and algorithm characteristics and platform-independent energy complexity model considering only algorithm characteristics. The platform-supporting model is used when platform parameters in the model are available while platform-independent model analyses energy complexity of algorithms without considering platform characteristics.
\subsection{Platform-supporting Energy Complexity Model} 
This section describes a methodology to find energy complexity of algorithms. The energy complexity model considers three groups of parameters: machine-dependent, algorithm-dependent and input-dependent parameters. The reason to consider all three parameter-categories is that only operational intensity \cite{Williams2009} is insufficient to capture the characteristics of algorithms. Two algorithms with the same values of operational intensity might consume different levels of energy. The reasons are their differences in data accessing patterns leading to performance scalability gap among them. For example, although the sequential version and parallel version of an algorithm may have the same operational intensity, they may have different energy consumption since the parallel version would have less static energy consumption because of shorter execution time.

The energy consumption of a parallel algorithm is the sum of i) static energy (or leakage) $E_{static}$, ii) dynamic energy of computation $E_{comp}$ and iii) dynamic energy of memory accesses $E_{mem}$: $E=E_{static}+E_{comp}+E_{mem}$ \cite{Choi2013, Korthikanti2009, Korthikanti2010}. The static energy $E_{static}$ is the product of the execution time of the algorithm and the static power of the whole platform. The dynamic energy of computation and the dynamic energy of memory accesses are proportional to the number of computational operations $Work$ and the number of memory accesses $I/O$, respectively. Pipelining technique in modern architectures enables overlapping computation with memory accesses \cite{HaTUTGRWA:2014}. Since computation time and memory-access time can be overlapped, the execution time of the algorithm is assumed to be the maximum of computation time and memory-access time \cite{Choi2013}. Therefore, the energy consumption of algorithms is computed by Equation \ref{eq:BigET}, 
where the values of ICE parameters, including $\epsilon_{op}$, $\epsilon_{I/O}$, $\pi_{op}$, and $\pi_{I/O}$ are described in Table \ref{table:ModelParameters} and computed by the Equation \ref{eq:epsilon_op}, \ref{eq:epsilon_IO}, \ref{eq:pi_op}, and \ref{eq:pi_IO}, respectively.  
\begin{equation} \label{eq:BigET}
	E= P^{sta} \times max(T^{comp} ,T^{mem}) + \epsilon_{op} \times Work + \epsilon_{I/O} \times I/O
\end{equation}
\comm{
\begin{equation} \label{eq:BigE0}
	E= \epsilon_{op} \times Work + \epsilon_{I/O} \times I/O + max(\pi_{op} \times Span ,\pi_{I/O} \times \frac{I/O \times Span}{Work})  
\end{equation}
}

\begin{table}
\caption{ICE Model Parameter Description}
\label{table:ModelParameters}
\begin{center}
\begin{tabular}{ll}
\hline\noalign{\smallskip}
Machine & Description \\
\noalign{\smallskip}\hline\noalign{\smallskip}
$\epsilon_{op}$                  & dynamic energy of one operation (average) \\ 
$\epsilon_{I/O}$                  & dynamic energy of a random memory access (1 core)\\ 
$\pi_{op}$                  & static energy when performing one operation  \\ 
$\pi_{I/O}$                  & static energy of a random memory access \\ 
\noalign{\smallskip}\hline\noalign{\smallskip}
\hline\noalign{\smallskip}
Algorithm & Description \\
\noalign{\smallskip}\hline\noalign{\smallskip}
$Work$                & Number of work in flops of the algorithm \cite{CormenLRS:2009}        \\ 
$Span$                 & The critical path of the algorithm \cite{CormenLRS:2009}       \\ 
$I/O$                  & Number of cache line transfer of the algorithm \cite{CormenLRS:2009}       \\ 
\noalign{\smallskip}\hline\noalign{\smallskip}
\end{tabular}
\end{center}
\end{table} 
 
\begin{equation} \label{eq:epsilon_op}
	\epsilon_{op}= P^{op} \times \frac{F}{Freq}
\end{equation}
\begin{equation} \label{eq:epsilon_IO}	
	\epsilon_{I/O}= P^{I/O} \times \frac{M}{Freq}
\end{equation}
\begin{equation} \label{eq:pi_op}
	\pi_{op}= P^{sta} \times \frac{F}{Freq}
\end{equation}
\begin{equation} \label{eq:pi_IO}
	\pi_{I/O}= P^{sta} \times \frac{M}{Freq}
\end{equation}

The dynamic energy of one operation by one core $\epsilon_{op}$ is the product of the consumed power of one operation by one active core $P^{op}$ and the time to perform one operation. Equation \ref{eq:epsilon_op} shows how $\epsilon_{op}$ relates to frequency $Freq$ and the number of cycles per operation $F$. Similarly, the dynamic energy of a random access by one core $\epsilon_{I/O}$ is the product of the consumed power by one active core performing one I/O (i.e., cache-line transfer) $P^{I/O}$ and the time to perform one cache line transfer computed as $M/Freq$, where $M$ is the number of cycles per cache line transfer (cf. Equation \ref{eq:epsilon_IO}). The static energy of operations $\pi_{op}$ is the product of the whole platform static power $P^{sta}$ and time per operation. The static energy of one I/O $\pi_{I/O}$ is the product of the whole platform static power and time per I/O, shown by Equation \ref{eq:pi_op} and \ref{eq:pi_IO}.

In order to compute {\em work}, {\em span} and {\em I/O} complexity of the algorithms, the input parameters also need to be considered. For example, SpMV algorithms consider input parameters listed in Table \ref{table:AlgorithmInputParameters}. Cache size is captured in the ICE model by the {\em I/O complexity} of the algorithm. Note that in the ICE machine model (Section \ref{EPEM-model}), cache size $Z$ is a constant and may disappear in the {\em I/O complexity} (e.g., O-notation).   

The details of how to obtain the ICE parameters of recent platforms are discussed in Section \ref{experiment-set-up}. The actual values of ICE platform parameters for 11 recent platforms are presented in Table \ref{table:platform-parameter-values}. 
\comm{
\begin{table}
\caption{Platform Parameter Description}
\label{table:PlatformParameters}
\begin{center}
\begin{tabular}{ll}
\hline\noalign{\smallskip}
Machine & Description \\
\noalign{\smallskip}\hline\noalign{\smallskip}
$P^{sta}$                & Static power of a whole platform         \\ 
$P^{op}$                  & Dynamic power of an operation        \\ 
$P^{I/O}$                  & Power to transfer one cache line \\
$N$                & Maximum number of cores in the platform      \\ 
$M$              & Number of cycles per cache line transfer \\ 
$F$              & Number of cycles per operation      \\
$Freq$                & Platform frequency       \\
$Z$              & Cache size of a single processor        
 \\
$B$              & Cache block size         \\ 
\noalign{\smallskip}\hline\noalign{\smallskip}
\end{tabular}
\end{center}
\end{table}
}
\begin{table*}
\caption{Platform parameter summary. The parameters of the first nine platforms are derived from \cite{Choi2014} and the parameters of the two new platforms are found in this study.}
\label{table:platform-parameter-values}
\begin{center}
\begin{adjustbox}{width=1\textwidth}
\begin{tabular}{llllll}
\hline\noalign{\smallskip}
Platform & Processor &  $\epsilon_{op}$(nJ) &$\pi_{op}$(nJ) &$\epsilon_{I/O}$(nJ) & $\pi_{I/O}$(nJ) \\
\noalign{\smallskip}\hline\noalign{\smallskip}
Nehalem i7-950	&Intel i7-950	&0.670 	&2.455	&50.88	&408.80\\
Ivy Bridge i3-3217U	&Intel  i3-3217U	&0.024 	&0.591	&26.75	&58.99\\
Bobcat CPU 	&AMD  E2-1800	&0.199 	&3.980	&27.84	&387.47\\
Fermi GTX 580	&NVIDIA GF100	&0.213 	&0.622	&32.83	&45.66\\
Kepler GTX 680	&NVIDIA GK104	&0.263 	&0.452	&27.97	&26.90\\
Kepler GTX Titan	&NVIDIA GK110	&0.094 	&0.077	&17.09	&32.94\\
XeonPhi KNC	&Intel 5110P	&0.012 	&0.178	&8.70	&63.65\\
Cortex-A9	&TI OMAP 4460	&0.302 &1.152	&51.84	&174.00\\
Arndale Cortex-A15	&Samsung Exynos 5	&0.275 	&1.385	&24.70	&89.34\\
\noalign{\smallskip}\hline\noalign{\smallskip}
Xeon 	&2xIntel E5-2650l v3	&0.263	&0.108	&8.86	&23.29\\
Xeon-Phi	&Intel 31S1P	&0.006	&0.078	&25.02	&64.40\\
\noalign{\smallskip}\hline\noalign{\smallskip}
\end{tabular}
\end{adjustbox}
\end{center}
\end{table*} 

The computation time of parallel algorithms is proportional to the span complexity of the algorithm, which is $T^{comp}=\frac{Span \times F}{Freq}$ where $Freq$ is the processor frequency, and $F$ is the number of cycles per operation. The memory-access time of parallel algorithms in the ICE model is proportional to the I/O complexity of the algorithm divided by its I/O parallelism, which is $T^{mem} = \frac{I/O}{I/O-parallelism} \times \frac{M}{Freq}$. As I/O parallelism, which is the average number of I/O ports that the algorithm can utilize per step along the span, is bounded by the computation parallelism $\frac{Work}{Span}$, namely the average number of cores that the algorithm can utilize per step along the span (cf. Section \ref{EPEM-model}), the memory-access time $T^{mem}$ becomes: $T^{mem}=\frac{I/O \times Span \times M}{Work \times Freq}$ where $M$ is the number of cycles per cache line transfer. If an algorithm has $T^{comp}$ greater than $T^{mem}$, the algorithm is a CPU-bound algorithm. Otherwise, it is a memory-bound algorithm. 
\subsubsection{CPU-bound Algorithms}
If an algorithm has computation time $T^{comp}$ longer than data-accessing time $T^{mem}$ (i.e., CPU-bound algorithms), the ICE energy complexity model becomes Equation \ref{eq:BigE-cpu-time} which is simplified as Equation \ref{eq:BigE1}.
\begin{equation} \label{eq:BigE-cpu-time}
	E= P^{sta} \times \frac{Span \times F}{Freq} + \epsilon_{op} \times Work + \epsilon_{I/O} \times I/O \\
\end{equation}
or
\begin{equation} \label{eq:BigE1}
	E= \pi_{op} \times Span + \epsilon_{op} \times Work + \epsilon_{I/O} \times I/O \\
\end{equation}

\subsubsection{Memory-bound Algorithms}
If an algorithm has data-accessing time longer than computation time (i.e., memory-bound algorithms): $T^{mem} \geq T^{comp}$, energy complexity becomes Equation \ref{eq:BigE-mem-time} which is simplified as Equation \ref{eq:BigE2}. 
\begin{equation} \label{eq:BigE-mem-time}
	E= P^{sta} \times \frac{I/O \times Span \times M}{Work \times Freq} + \epsilon_{op} \times Work + \epsilon_{I/O} \times I/O \\
\end{equation}
or
\begin{equation} \label{eq:BigE2}
	E= \pi_{I/O} \times \frac{I/O \times Span}{Work} +  \epsilon_{op} \times Work + \epsilon_{I/O} \times I/O \\
\end{equation}
\subsection{Platform-independent Energy Complexity Model}
This section describes the energy complexity model that is platform-independent and considers only algorithm characteristics. This complexity model is used when analyzing energy complexity of an algorithm without platform parameters. When the platform parameters (i.e., $\epsilon_{op}$, $\epsilon_{I/O}$, $\pi_{op}$, and $\pi_{I/O}$) are unavailable, the energy complexity model is derived from Equation \ref{eq:BigET} because the platform parameters are constants and can be removed. Assuming $\pi_{max} = max(\pi_{op}, \pi_{I/O})$, after removing platform parameters, the platform-independent energy complexity model are shown in Equation \ref{eq:BigE-non-flatform}.
\begin{equation} \label{eq:BigE-non-flatform}
	E= O(Work+I/O+max(Span, \frac{I/O \times Span}{Work}))  
\end{equation} 
\section{A Case Study of Sparse Matrix Multiplication}
\label{SpMV-energy}
SpMV is one of the most common application kernels in Berkeley dwarf list \cite{Asa06}. It computes a vector result $y$ by multiplying a sparse matrix $A$ with a dense vector $x$: $y=Ax$. SpMV is a data-intensive kernel and has irregular memory-access patterns. The data access patterns for SpMV is defined by its sparse matrix format and matrix input types. 
There are several sparse matrix formats and SpMV algorithms in literature. To name a few, they are Coordinate Format (COO), Compressed Sparse Column (CSC), Compressed Sparse Row (CSR), Compressed Sparse Block (CSB), Recursive Sparse Block (RSB), Block Compressed Sparse Row (BCSR) and so on.
Three popular SpMV algorithms, namely CSC, CSB and CSR are chosen to validate the proposed energy complexity model. They have different data-accessing patterns leading to different values of I/O, work and span complexity. Since SpMV is a memory-bound application kernel, Equation \ref{eq:BigE2} is applied. The input matrices of SpMV have different parameters listed in Table \ref{table:AlgorithmInputParameters}.
\begin{table}
\caption{SpMV Input Parameter Description}
\label{table:AlgorithmInputParameters}
\begin{center}
\begin{tabular}{ll}
\hline\noalign{\smallskip}
SpMV Input & Description \\
\noalign{\smallskip}\hline\noalign{\smallskip}
$n$                & Number of rows        \\ 
$nz$                 & Number of nonzero elements        \\ 
$nr$                  & Maximum number of nonzero in a row        \\ 
$nc$                  & Maximum number of nonzero in a column        \\ 
$\beta$                  & Size of a block \\ 
\noalign{\smallskip}\hline\noalign{\smallskip}
\end{tabular}
\end{center}
\end{table}
\subsection{Compressed Sparse Row}
CSR is a standard storage format for sparse matrices which reduces the storage of matrix compared to the tuple representation \cite{Kotlyar:1997}. This format enables row-wise compression of $A$ with size $n \times n$ (or  $n \times m$) to store only the non-zero $nz$ elements. Let {\em nz} be the number of non-zero elements in matrix A. 
The {\em work} complexity of CSR SpMV is $\Theta(nz)$ where $nz>=n$ and {\em span} complexity is $O(nr + \log{n})$ \cite{Buluc:2009}, where $nr$ is the maximum number of non-zero elements in a row. The {\em I/O} complexity of CSR in the sequential I/O model of row-major layout is $O(nz)$ \cite{Bender2010} namely, scanning all non-zero elements of matrix $A$ costs $O(\frac{nz}{B})$ I/Os with B is the cache block size. However, randomly accessing vector $x$ causes the total of $O(nz)$ I/Os.
Applying the proposed model on CSR SpMV, their total energy complexity are computed as Equation \ref{eq:BigE-CSR}.
\begin{equation} \label{eq:BigE-CSR}
	E_{CSR}= O(\epsilon_{op} \times nz + \epsilon_{I/O} \times nz + \pi_{I/O} \times (nr+ \log{}n)) 
\end{equation}
\subsection{Compressed Sparse Column}
CSC is the similar storage format for sparse matrices as CSR. However, it compresses the sparse matrix in column-wise manner to store the non-zero elements. The {\em work} complexity of CSC SpMV is $\Theta(nz)$ where $nz>=n$ and {\em span} complexity is $O(nc + \log{n})$, where $nc$ is the maximum number of non-zero elements in a column. The {\em I/O} complexity of CSC in the sequential I/O model of column-major layout is $O(nz)$ \cite{Bender2010}. Similar to CSR, scanning all non-zero elements of matrix $A$ in CSC format costs $O(\frac{nz}{B})$ I/Os. However, randomly updating vector $y$ causing the bottle neck with total of $O(nz)$ I/Os.
Applying the proposed model on CSC SpMV, their total energy complexity are computed as Equation \ref{eq:BigE-CSC}.
\begin{equation} \label{eq:BigE-CSC}
	E_{CSC}= O(\epsilon_{op} \times nz + \epsilon_{I/O} \times nz + \pi_{I/O} \times (nc+ \log{}n)) 
\end{equation}
\subsection{Compressed Sparse Block}
Given a sparse matrix $A$, while CSR has good performance on SpMV $y=Ax$, CSC has good performance on transpose sparse matrix vector multiplication $y=A^{T}\times x$, Compressed sparse blocks (CSB) format is efficient for computing either $Ax$ or $A^{T}x$. CSB is another storage format for representing sparse matrices by dividing the matrix $A$ and vector $x, y$  to blocks. A block-row contains multiple chunks, each chunks contains consecutive blocks and non-zero elements of each block are stored in Z-Morton-ordered \cite{Buluc:2009}.
From Beluc et al. \cite{Buluc:2009}, CSB SpMV computing a matrix with $nz$ non-zero elements, size $n\times n$ and divided by block size $\beta \times \beta$ has span complexity $O(\beta \times \log{\frac{n}{\beta}}+ \frac{n}{\beta})$ and {\em work} complexity as $\Theta(\frac{n^2}{\beta^2}+nz)$.

{\em I/O} complexity for CSB SpMV is not available in the literature. We do the analysis of CSB manually by following the master method \cite{CormenLRS:2009}. The {\em I/O} complexity is analyzed for the algorithm CSB\_SpMV(A,x,y) from Beluc et al. \cite{Buluc:2009}. The I/O complexity of CSB is similar to {\em work} complexity of CSB $O(\frac{n^2}{\beta^2} + nz)$, only that non-zero accesses in a block is divided by B: $O(\frac{n^2}{\beta^2} + {\frac{nz}{B}})$, where $B$ is cache block size. The reason is that non-zero elements in a block are stored in Z-Morton order which only requires $\frac{nz}{B}$ I/Os. The energy complexity of CSB SPMV is shown in Equation \ref{eq:BigE-CSB}.

From the complexity analysis of SpMV algorithms using different layouts, the complexity of CSR-SpMV, CSC-SpMV and CSB-SpMV are summarized in Table \ref{table:SpMV-complexity}.
\begin{floatEq}
\begin{equation}
\label{eq:BigE-CSB}
	E_{CSB}= O(\epsilon_{op} \times (\frac{n^2}{\beta^2} + nz) + \epsilon_{I/O} \times (\frac{n^2}{\beta^2} + \frac{nz}{B}) + \pi_{I/O} \times \frac{(\frac{n^2}{\beta^2} + \frac{nz}{B})\times (\beta \times \log{\frac{n}{\beta}}+ \frac{n}{\beta})}{(\frac{n^2}{\beta^2} + nz)} )
\end{equation}
\end{floatEq}
\begin{table*}
\caption{SpMV Complexity Analysis}
\label{table:SpMV-complexity}
\begin{center}
\begin{tabular}{llll}
\hline\noalign{\smallskip}
Complexity & CSC-SpMV & CSB-SpMV & CSR-SpMV \\
\noalign{\smallskip}\hline\noalign{\smallskip}
Work  & $\Theta(nz)$ \cite{Buluc:2009} & $\Theta(\frac{n^2}{\beta^2} + nz)$ \cite{Buluc:2009} & $\Theta(nz)$ \cite{Buluc:2009}\\ 
I/O                 & $O(nz)$ \cite{Bender2010} & $O(\frac{n^2}{\beta^2} + {\frac{nz}{B}})$ [this study] & $O(nz)$ \cite{Bender2010} \\ 
Span                  & $O(nc+ \log{}n)$ \cite{Buluc:2009} & $O(\beta \times \log{\frac{n}{\beta}}+ \frac{n}{\beta})$ \cite{Buluc:2009} & $O(nr+ \log{}n)$  \cite{Buluc:2009}\\ 

\noalign{\smallskip}\hline\noalign{\smallskip}
\end{tabular}
\end{center}
\end{table*}

\section{A Case Study of Dense Matrix Multiplication}
\label{Matmul-energy}
Besides SpMV, we also apply the ICE model to dense matrix multiplication (matmul). Unlike SpMV, a data-intensive kernel, matmul is a computation-intensive kernel used in high performance computing. It computes output matrix C (size n x p) by multiplying two dense matrices A (size n x m) and B (size m x p): $C=A \times B$. In this work, we implemented two matmul algorithms (i.e., a basic algorithm and a cache-oblivious algorithm \cite{FrigoLPR:1999}) and apply the ICE analysis to find their energy complexity. Both algorithms partition matrix A and C equally to N sub-matrices (e.g., $A_{i}$ with i=(1,2,..,N)), where N is the number of cores in the platform. The partition approach is shown in Figure \ref{fig:matmul-algo}. Each core computes a sub-matrix $C_{i}$: $C_{i}=A_{i} \times B$. Since matmul is a computation-bound application kernel, Equation \ref{eq:BigE1} is applied.
\begin{figure}[!t] \centering
\resizebox{0.6\columnwidth}{!}{ \includegraphics{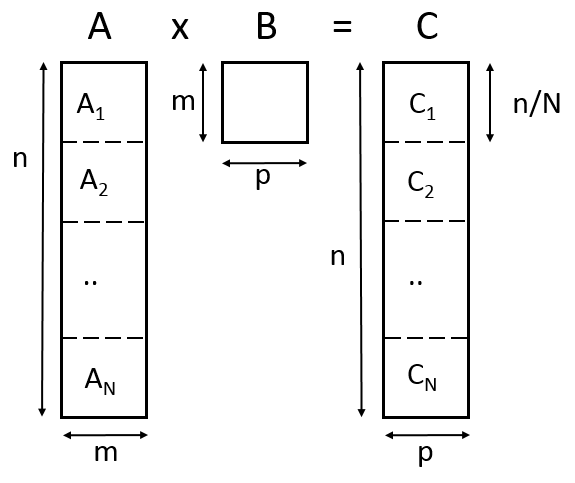}}
\caption{Partition approach for parallel matmul algorithms. Each sub-matrix $A_{i}$ has size $\frac{n}{N} \times m$ and each sub-matrix $C_{i}$ has size $\frac{n}{N} \times p$.}
\label{fig:matmul-algo}
\end{figure} 
\subsection{Basic Matmul Algorithm}
The basic matmul algorithm is described in Figure \ref{fig:matmul-simple}. Its work complexity is $\Theta(2nmp)$ \cite{Yelick2004} and span complexity is $\Theta(\frac{2nmp}{N})$ because the computational work is divided equally to N cores due to matrix partition approach. When matrix size of matrix B is bigger than platform cache size, the basic algorithm loads matrix B n times (i.e., once for computing each row of C), results in $\frac{nmp}{B}$ cache block transfer, where $B$ is cache block size. In total, I/O complexity of the basic matmul algorithm is $\Theta(\frac{nm+nmp+np}{B})$. Applying the ICE model on this algorithm, the total energy complexity is computed as Equation \ref{eq:BigE-Naive}.
\begin{figure}[!t] \centering
\resizebox{0.8\columnwidth}{!}{ \includegraphics{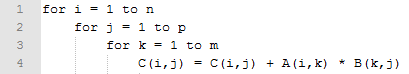}}
\caption{Basic matmul algorithm, where sizes of matrix A, B, C are nxm, mxp, nxp, respectively.}
\label{fig:matmul-simple}
\end{figure} 
\begin{equation} \label{eq:BigE-Naive}
	E_{basic}= O(\epsilon_{op} \times 2nmp + \epsilon_{I/O} \times \frac{nm+nmp+np}{B} + \pi_{op} \times \frac{2nmp}{N}) 
\end{equation}
\subsection{Cache-oblivious Matmul Algorithm}
The cache-oblivious matmul (CO-matmul) algorithm \cite{FrigoLPR:1999} is a divide-and-conquer algorithm. It has work complexity the same as the basic matmul algorithm $\Theta(2nmp)$. Its span complexity is also $\Theta(\frac{2nmp}{N})$ because of the used matrix partition approach shown in Figure \ref{fig:matmul-algo}. The I/O complexity of CO-matmul, however, is different from the basic algorithm: $\Theta(n+m+p+\frac{nm+mp+np}{B} +\frac{nmp}{B\sqrt[2]{Z}})$ \cite{FrigoLPR:1999}. Applying the ICE model to CO-matmul, the total energy complexity is computed as Equation \ref{eq:BigE-CO}.
\begin{floatEq}
\begin{equation} \label{eq:BigE-CO}
	E_{CO}= O(\epsilon_{op} \times 2nmp + \epsilon_{I/O} \times (n+m+p+\frac{nm+mp+np}{B}+\frac{nmp}{B\sqrt[2]{Z}})+\pi_{op} \times \frac{2nmp}{N})
\end{equation}
\end{floatEq}
\begin{table*}
\caption{Matmul Complexity Analysis}
\label{table:Matmul-complexity}
\begin{center}
\begin{adjustbox}{width=1\textwidth}
\begin{tabular}{llll}
\hline\noalign{\smallskip}
Complexity & Cache-oblivious Algorithm & Basic Algorithm\\
\noalign{\smallskip}\hline\noalign{\smallskip}
Work  & $\Theta(2nmp)$ \cite{FrigoLPR:1999}  & $\Theta(2nmp)$  \cite{Yelick2004}  \\ 
I/O   & $\Theta(n+m+p+\frac{nm+mp+np}{B} +\frac{nmp}{B\sqrt[2]{Z}})$ \cite{FrigoLPR:1999}  & $\Theta(\frac{nm+nmp+np}{B})$ [this study]\\ 
Span  & $\Theta(\frac{2nmp}{N})$ [this study] & $\Theta(\frac{2nmp}{N})$ [this study]\\ 
\noalign{\smallskip}\hline\noalign{\smallskip}
\end{tabular}
\end{adjustbox}
\end{center}
\end{table*}

\section{Validation of ICE Model}
\label{validation}
This section describes the experimental study to validate the ICE model, including: introducing the two experimental platforms and how to obtain their parameters for the ICE model (cf. Section \ref{experiment-set-up}), describing SpMV implementation and sparse matrix types used in this validation (cf. Section \ref{SpMV-Implementation}), and discussing the validation results of SpMV and matmul.
\subsection{Experiment Set-up}
\label{experiment-set-up}
For the validation of the ICE model, we conduct the experiments on two HPC platforms: one platform with two Intel Xeon E5-2650l v3 processors and one platform with Xeon Phi 31S1P processor. The Intel Xeon platform has two processors Xeon E5-2650l v3 with $2\times12$ cores, each processor has the frequency 1.8 GHz. The Intel Xeon Phi platform has one processor Xeon Phi 31S1P with $57$ cores and its frequency is 1.1 GHz. To measure energy consumption of the platforms, we read the PCM MSR counters for Intel Xeon and MIC power reader for Xeon Phi.
\subsection{Identifying Platform Parameters}
We apply the energy roofline approach \cite{Choi2013, Choi2014} to find the platform parameters for the two new experimental platforms, namely Intel Xeon E5-2650l v3 and Xeon Phi 31S1P. Moreover, the energy roofline study \cite{Choi2014} has also provided a list of other platforms including CPU, GPU, embedded platforms with their parameters considered in the Roofline model. Thanks to authors Choi et al. \cite{Choi2014}, we extract the required values of ICE parameters for nine platforms presented in their study as follows: $\epsilon_{op}=\epsilon_{d}$, $\epsilon_{I/O}=\epsilon_{mem}\times B$, $\pi_{op}=\pi_{1} \times \tau_{d}$, $\pi_{I/O}=\pi_{1} \times \tau_{mem}$, where $B$ is cache block size, $\epsilon_{d}$, $\epsilon_{d}$, $\tau_{d}$, $\tau_{mem}$ are defined by \cite{Choi2014} as energy per flop, energy per byte, time per flop and time per byte, respectively. 

The ICE parameter values of the two new HPC platforms (i.e., Xeon and Xeon-Phi 31S1P) used to validate the ICE model are obtained by using the same approach as energy roofline study \cite{Choi2013}. We create micro-benchmarks for the two platforms and measure their energy consumption and performance. The ICE parameter values of each platform are obtained from energy and performance data by regression techniques. Along with the two HPC platforms used in this validation, we provide parameters required in the ICE model for a total of 11 platforms. Their platform parameters are listed in Table \ref{table:platform-parameter-values} for further uses.
\subsection{SpMV Implementation}
\label{SpMV-Implementation}
We want to conduct complexity analysis and experimental study with two SpMV algorithms, namely CSB and CSC. Parallel CSB and sequential CSC implementations are available thanks to the study from Bulu\c{c} et al. \cite{Buluc:2009}. Since the optimization steps of available parallel SpMV kernels (e.g., pOSKI \cite{pOSKI}, LAMA\cite{Forster2011}) might affect the work complexity of the algorithms, we decided to implement a simple parallel CSC using Cilk and pthread. To validate the correctness of our parallel CSC implementation, we compare the vector result $y$ from $y=A*x$ of CSC and CSB implementation. The comparison shows the equality of the two vector results $y$. Moreover, we compare the performance of the our parallel CSC code with Matlab parallel CSC-SpMV kernel. Matlab also uses CSC layout as the format for their sparse matrix \cite{Gilbert:1992} and is used as baseline comparison for SpMV studies \cite{Buluc:2009}. Our CSC implementation has out-performed Matlab parallel CSC kernel when computing the same targeted input matrices at least 136\% across different inputs from Table \ref{table:matrix-type}. 
The experimental study of SpMV energy consumption is then conducted with CSB SpMV implementation from Bulu\c{c} et al. \cite{Buluc:2009} and our CSC SpMV parallel implementation.  
\begin{figure}[!t] \centering
\resizebox{0.8\columnwidth}{!}{ \includegraphics{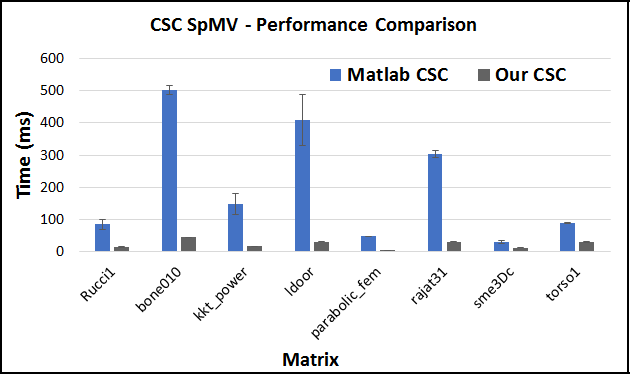}}
\caption{Performance (time) comparison of two parallel CSC SpMV implementations. For a set of different input matrices, the parallel CSC SpMV using Cilk out-performs Matlab parallel CSC.}
\label{fig:Matlab-comparison}
\end{figure}
\subsection{SpMV Matrix Input Types}
We conducted the experiments with nine different matrix-input types from Florida sparse matrix collection \cite{Davis:2011}. Each matrix input has different properties listed in Table \ref{table:AlgorithmInputParameters}, including size of the matrix $n\times m$, the maximum number of non-zero of the sparse matrix $nz$, the maximum number of non-zero elements in one column $nc$. Table \ref{table:matrix-type} lists the matrix types used in this experimental validation with their properties. 
\begin{table}
\caption{Sparse matrix input types. The maximum number of non-zero elements in a column $nc$ is derived from \cite{Buluc:2009}.}
\label{table:matrix-type}
\begin{center}
\begin{tabular}{lllll}
\hline\noalign{\smallskip}
\textbf{Matrix type}	&\textbf{n}	&\textbf{m}	&\textbf{nz}	&\textbf{nc}	\\
\noalign{\smallskip}\hline\noalign{\smallskip}
bone010	&986703	&986703	&47851783	&63	\\
kkt\_power	&2063494	&2063494	&12771361	&90	\\
ldoor	&952203	&952203	&42493817	&77	\\
parabolic\_fem	&525825	&525825	&3674625	&7	\\
pds-100	&156243	&517577	&1096002	&7	\\
rajat31	&4690002	&4690002	&20316253	&1200	\\
Rucci1	&1977885	&109900	&7791168	&108	\\
sme3Dc	&42930	&42930	&3148656	&405	\\
torso1	&116158	&116158	&8516500	&1200	\\
\noalign{\smallskip}\hline\noalign{\smallskip}
\end{tabular}
\end{center}
\end{table}
\subsection{Validating ICE Using Different SpMV Algorithms}
The model aims to compare energy consumption of two algorithm. Therefore, we validate the ICE model by showing the comparison using the ratio of energy consumption of two algorithms. From the model-estimated data, CSB SpMV consumes less energy than CSC SpMV on both platforms. Even though CSB has higher work complexity than CSC, CSB SpMV has less I/O complexity than CSC SpMV. Firstly, the dynamic energy cost of one I/O is much greater than the energy cost of one operation (i.e., $\epsilon_{I/O}>>\epsilon_{op}$) on both platforms. Secondly, CSB has better parallelism than CSC, computed by $\frac{Work}{Span}$, which results in shorter execution time. Both reasons contribute to the less energy consumption of CSB SpMV. 
\begin{figure}[!t] \centering
\resizebox{0.9\columnwidth}{!}{ \includegraphics{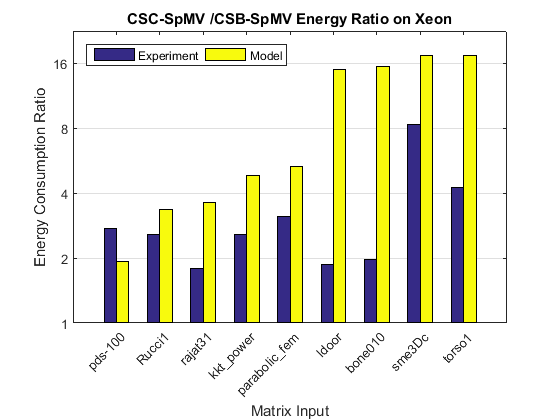}}
\caption{Energy consumption comparison between CSC-SpMV and CSB-SpMV on the Intel Xeon platform, computed by $\frac{E_{CSC}}{E_{CSB}}$. Both the ICE model estimation and experimental measurement on Intel Xeon platform show the consistent results that $\frac{E_{CSC}}{E_{CSB}}$ is greater than 1, meaning CSC SpMV algorithm consumes more energy than the CSB SpMV algorithm on different input matrices.}
\label{fig:CSBvsCSC-Xeon}
\end{figure} 
\begin{figure}[!t] \centering
\resizebox{0.9\columnwidth}{!}{ \includegraphics{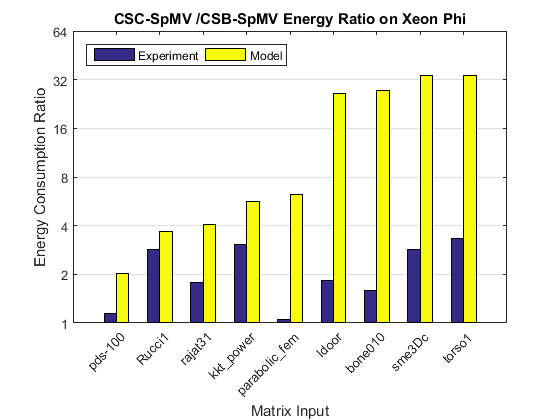}}
\caption{Energy consumption comparison between CSC-SpMV and CSB-SpMV on the Intel Xeon Phi platform, computed by $\frac{E_{CSC}}{E_{CSB}}$. Both the ICE model estimation and experimental measurement on Intel Xeon Phi platform show the consistent results that $\frac{E_{CSC}}{E_{CSB}}$ is greater than 1, meaning CSC SpMV algorithm consumes more energy than the CSB SpMV algorithm on different input matrices.}
\label{fig:CSBvsCSC-XeonPhi}
\end{figure} 
The measurement data confirms that CSB SpMV algorithm consumes less energy than CSC SpMV algorithm, shown by the energy consumption ratio between CSC-SpMV and CSB-SpMV greater than 1 in the Figure \ref{fig:CSBvsCSC-Xeon} and \ref{fig:CSBvsCSC-XeonPhi}. For all input matrices, the ICE model has confirmed that CSB SpMV consumes less energy than CSC SpMV algorithm. Because the model has abstracted possible platform by only 4 parameters (i.e., $\epsilon_{op}$, $\epsilon_{I/O}$, $\pi_{op}$, and $\pi_{I/O}$), there are the differences between the model and experiment ratios shown in the Figure \ref{fig:CSBvsCSC-Xeon} and \ref{fig:CSBvsCSC-XeonPhi}. For accurate models that provide the precise energy estimation, the platform parameters need to be highly detailed such as RTHpower model for embedded platforms \cite{Tran2016, TranSamos2016}.
\comm{
\subsection{Validating The Applicability of ICE on Different Platforms}
\begin{figure*}[!t] \centering
\resizebox{0.9\textwidth}{!}{\includegraphics{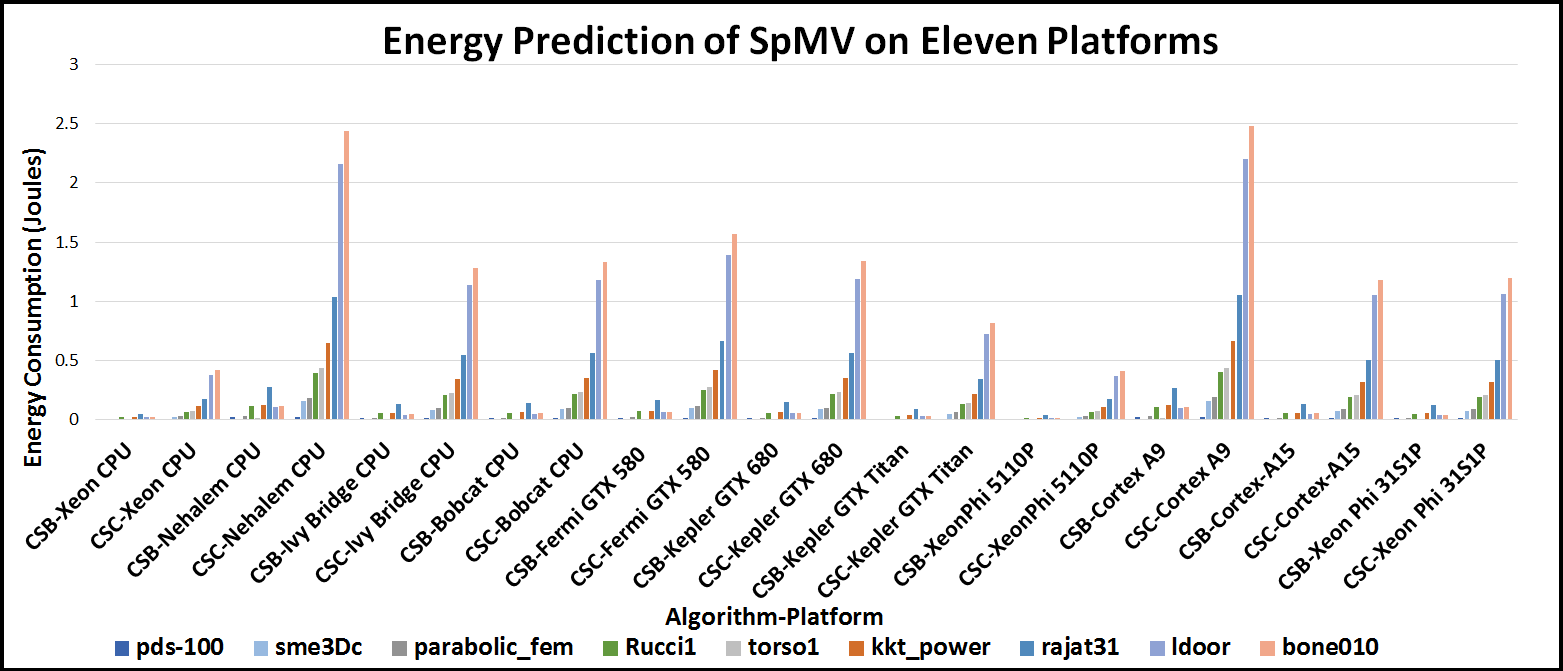}}
\caption{Energy Comparison of CSB and CSC SpMV on eleven different platforms.}
\label{fig:ModelPrediction_Platforms}
\end{figure*}
The energy comparison of CSB and CSC SpMV is concluded for eleven platforms listed in Table \ref{table:platform-parameter-values}. Like two Xeon and Xeon Phi 31S1P platforms used in experiments, Figure \ref{fig:ModelPrediction_Platforms} shows the prediction that CSB SpMV consumes less energy than CSC SpMV, on all platforms listed in Table \ref{table:platform-parameter-values}. This confirms the applicability of ICE model to compare the energy consumption of algorithms on different platforms with different input types.
}
\label{SpMV-validation}
\subsection{Validating ICE With Matmul Algorithms}
From the model-estimated data, Basic-Matmul consumes more energy than CO-Matmul on both platforms. Even though both algorithms have the same work and span complexity, Basic-Matmul has more I/O complexity than CO-Matmul, which results in greater energy consumption of Basic-Matmul compared to CO-Matmul algorithm.
\begin{figure}[!t] \centering
\resizebox{0.8\columnwidth}{!}{ \includegraphics{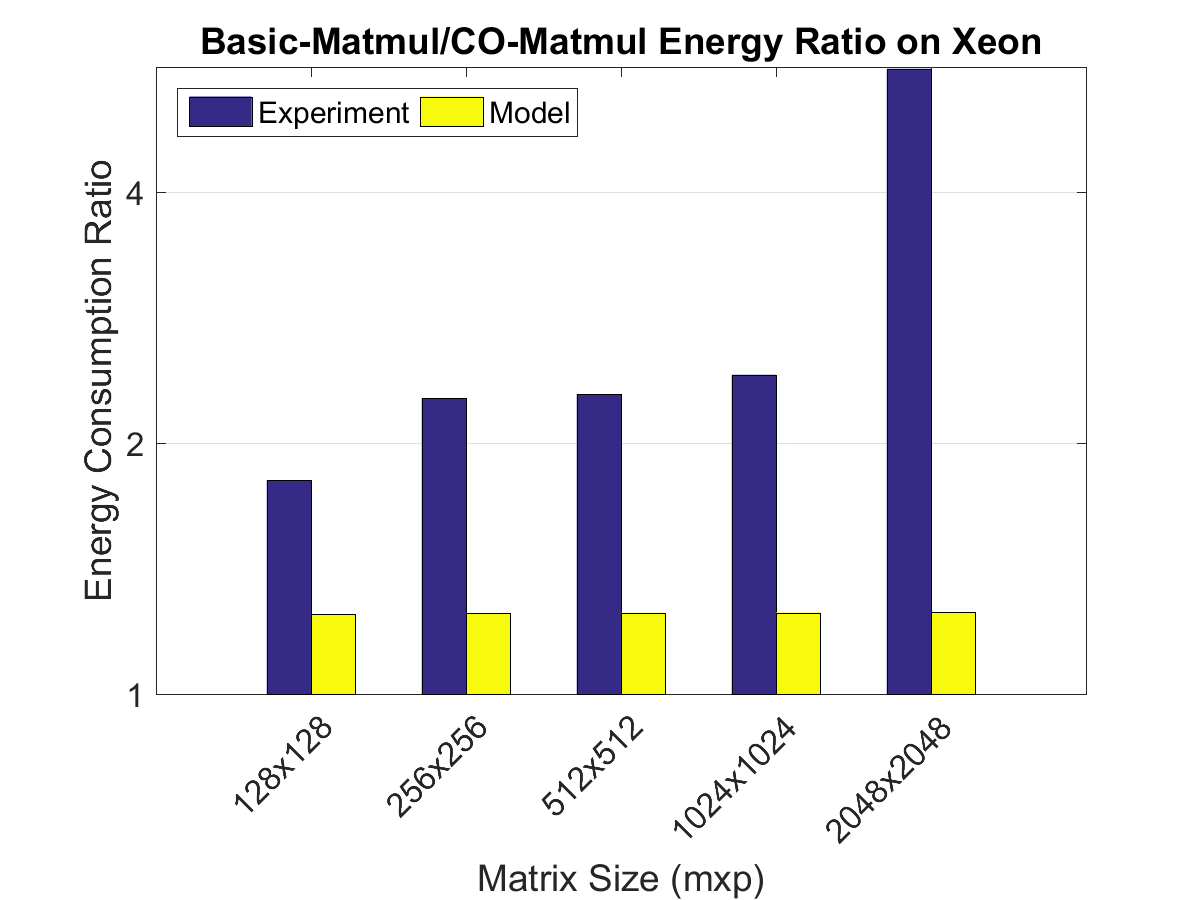}}
\caption{Energy consumption comparison between Basic-Matmul and CO-Matmul on the Intel Xeon platform, computed by $\frac{E_{Basic}}{E_{CO}}$. Both the ICE model estimation and experimental measurement on Intel Xeon platform show that $\frac{E_{Basic}}{E_{CO}}$ is greater than 1, meaning Basic-Matmul algorithm consumes more energy than the CO-Matmul algorithm.}
\label{fig:NaiveCO-Xeon}
\end{figure} 
\begin{figure}[!t] \centering
\resizebox{0.8\columnwidth}{!}{ \includegraphics{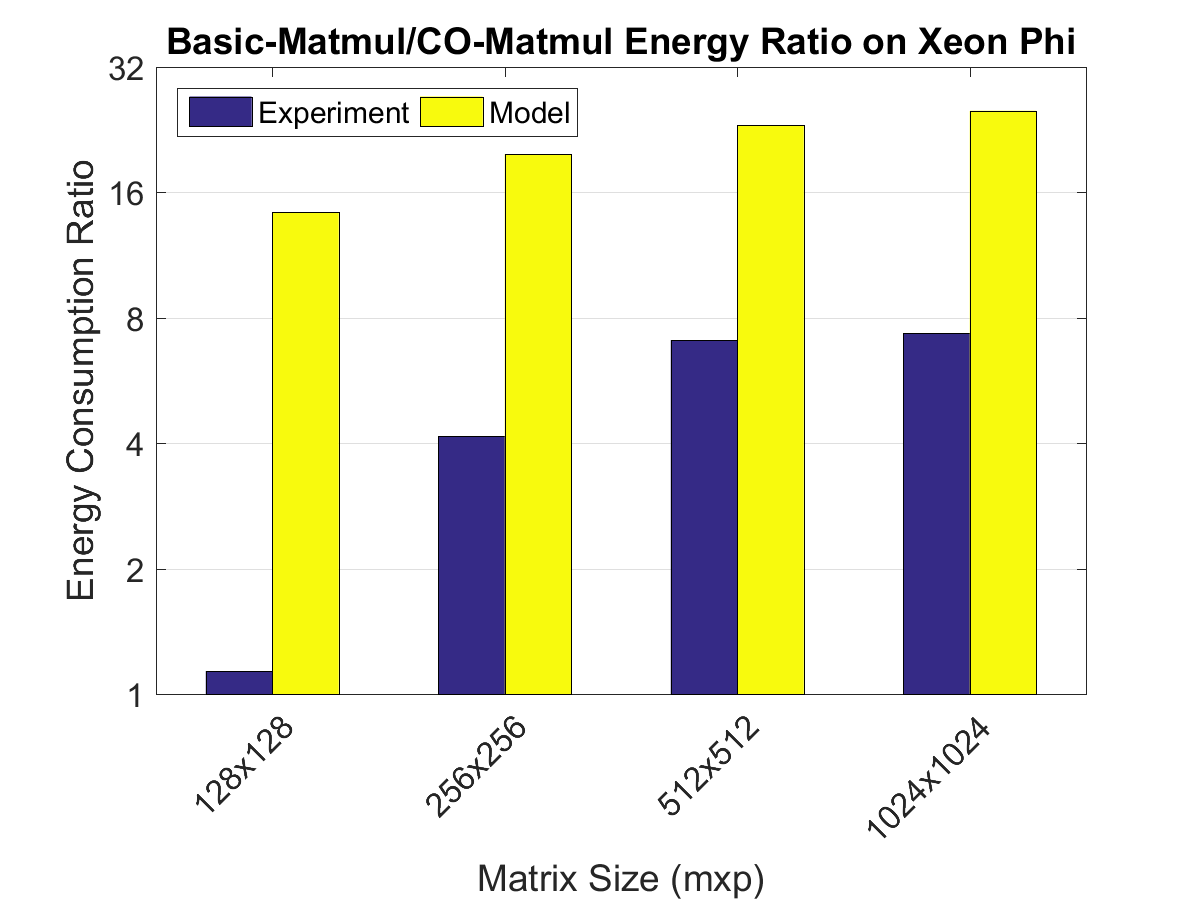}}
\caption{Energy consumption comparison between Basic-Matmul and CO-Matmul on the Intel Xeon Phi platform, computed by $\frac{E_{Basic}}{E_{CO}}$. Both the ICE model estimation and experimental measurement on Intel Xeon Phi platform show that $\frac{E_{Basic}}{E_{CO}}$ is greater than 1, meaning Basic-Matmul algorithm consumes more energy than the CO-Matmul algorithm.}
\label{fig:NaiveCO-XeonPhi}
\end{figure} 
The measurement data confirms that Basic-Matmul algorithm consumes more energy than CO-Matmul algorithm, shown by the energy consumption ratio between Basic-Matmul and CO-Matmul greater than 1 in the Figure \ref{fig:NaiveCO-Xeon} and \ref{fig:NaiveCO-XeonPhi}. For all input matrices, the ICE model has confirmed that Basic-Matmul consumes more energy than CO-Matmul algorithm.

\label{Matmul-validation}

\section{Related Work - Overview of energy models}
\label{related-work}
 
\comm{
We present the summary of existing modeling studies in Table \ref{table:energy-model-summary}. The characteristics of each approaches are extracted as the list of categories, including: whether the models support parallel algorithms (i.e., Parallel Algorithm Support), whether the model is applicable to general algorithms (i.e., Applicability), whether the model is validated (i.e., Validation), the communication model (i.e., Communication model), whether there is pre-run overhead before estimating energy consumption of applications (i.e., Pre-run Overhead) and how the model represents applications (i.e., App-properties). This summary is not an exhaustive survey on the topic of energy models. However, we believe the Table \ref{table:energy-model-details} represents the most current studies on energy models.
}
Energy models for finding energy-optimized system configurations for a given application have been recently reported [12, 16, 19]. Imes et al. \cite{7108419} used controller theory and linear programming to find energy-optimized configurations for an application with soft real-time constraints at runtime. 
Mishra et al. \cite{Mishra:2015} used hierarchical Bayesian model in machine learning to find  energy-optimized configurations. 
Snowdon et al. \cite{Snowdon:2009} developed a power management framework called Koala which models the energy consumption of the platform and monitors an application' energy behavior. 
Although the energy models for finding energy-optimized system configurations have resulted in energy saving in practice, they focus on characterizing system platforms rather than applications and therefore are not appropriate for analyzing the energy complexity of application algorithms. 

Another direction of energy modeling study is to predict the energy consumption of applications by analyzing applications without actual execution on real platforms which we classify as analytic models. 

Energy roofline models \cite{Choi2013, Choi2014} are some of the comprehensive energy models that abstract away possible algorithms in order to analyze and characterize different multicore platforms in terms of energy consumption. 
Our new energy model, which abstracts away possible multicore platform and characterize the energy complexity of algorithms based on their {\em work, span} and {\em I/O} complexity, complements the energy roofline models.  

Validated energy models for {\em specific} algorithms have been reported recently \cite{Alonso2014, Malossi2015}. Alonso et al. \cite{Alonso2014} provided an accurate energy model for three key dense matrix factorizations. Malossi et al. \cite{Malossi2015} focused on basic linear-algebra kernels and characterized the kernels by the number of arithmetic operations, memory accesses, reduction and barrier steps. Although the energy models for specific algorithms are accurate for the target algorithms, they are not applicable for other algorithms and therefore cannot be used as general energy complexity models for parallel algorithms.

The {\em energy scalability} of a parallel algorithm has been investigated by Korthikanti et al. \cite{Korthikanti2009, Korthikanti2010}. 
Unlike the energy scalability studies that have not been validated on real platforms, our new energy complexity model is validated on HPC and accelerator platforms, confirming its usability and accuracy.

The energy complexity of {\em sequential} algorithms on a {\em uniprocessor} machine with {\em several memory banks} has been studied by Roy et al. \cite{Roy2013}. Our energy complexity studies complement Roy et al.'s studies by investigating the energy complexity of {\em parallel} algorithms on a {\em multiprocessor} machine with {\em a shared memory bank} and private caches, a machine model that has been widely adopted to study parallel algorithms \cite{Frigo:2006, Arge:2008, Korthikanti2010}.

\section{Conclusion}
In this study, we have devised a new general model for analyzing the energy complexity of multithreaded algorithms. The energy complexity of an algorithm is derived from its \textit{work}, \textit{span} and \textit{I/O} complexity.
Moreover, two case studies are conducted to demonstrate how to use the model to analyze the energy complexity of SpMV algorithms and matmul algorithms. The energy complexity analyses are validated for two SpMV algorithms and two matmul algorithm on two HPC platforms with different input matrices. The experimental results confirm the theoretical analysis with respect to which algorithm consumes more energy. The ICE energy complexity model gives algorithm-developers the insight into which algorithm is analytically more energy-efficient. Improving the ICE model by considering the numbers of platform cores is a part of our future work.
\comm{
In the future, we would extend our work in two directions:
\begin{itemize}
\item We want to develop a run-time framework which can choose the least energy-consuming implementations among the available kernels at run-time using the proposed energy model.
\item Nowadays, there are executable frameworks that connect different platforms to one task scheduler. Selecting the most energy-efficient platforms or system configurations to run applications is one of the techniques to achieve energy optimization. 
In order to do so, the energy models need to be able to model the details of each platform. We want to improve the energy models to compare the energy consumption of applications between different platforms. 
\end{itemize}
}

\bibliographystyle{IEEEtranS}
\bibliography{PACT} 

\end{document}